\documentclass[conference,10pt]{IEEEtran}

\newif \ifconf
\conftrue

\usepackage[utf8]{inputenc} 
\usepackage[T1]{fontenc}
\usepackage{url}
\usepackage{ifthen}
\usepackage{cite}
\usepackage[cmex10]{amsmath} 
\usepackage{adjustbox}
\usepackage{graphicx} 
\usepackage{hyperref}
\hypersetup{colorlinks=true, unicode=true, linkcolor=[rgb]{0.10,0.05,0.67}, citecolor=[rgb]{0.10,0.05,0.67}, filecolor=[rgb]{0.10,0.05,0.67}, urlcolor=[rgb]{0.10,0.05,0.67}}
\usepackage{physics}
\usepackage{setspace}
\usepackage{amssymb}
\usepackage{stfloats}
\usepackage{amsthm}
\usepackage{multirow}
\usepackage{amsfonts}
\usepackage{xcolor}
\usepackage{graphicx}
\usepackage{subcaption}
\usepackage[]{algorithmicx}
\usepackage{algpseudocode,algorithm}
\usepackage{mathtools}
\usepackage{stmaryrd}
\usepackage[mathscr]{euscript}
\usepackage{array}
\usepackage{mathrsfs}

\usepackage{soul}

\newtheorem{theorem}{Theorem}
\newtheorem{lemma}{Lemma}
\newtheorem{cor}{Corollary}

\usepackage{authblk}
\usepackage{tabularx}
\usepackage{svg}
\usepackage{tikz}
\usepackage{bm}
\usetikzlibrary{arrows.meta}
\usepackage{layouts}

\newcommand{\off}[1]{}

\begin{document}
\title{\hspace{-0.02cm} Adaptive Compressed Integrate-and-Fire Time Encoding Machine\vspace{-0.4cm}}

\author{Vered Karp, Aseel Omar and Alejandro Cohen\vspace{0.0cm}\\
Faculty of ECE, Technion, Israel,\\ Emails: \{veredlevi,aseel.omar\}@campus.technion.ac.il and alecohen@technion.ac.il \vspace{-0.59cm}}


\maketitle


\begin{abstract} 
Integrate-and-Fire Time Encoding Machine (IF-TEM) is a power-efficient asynchronous sampler that converts analog signals into non-uniform time-domain samples. Adaptive IF-TEM (AIF-TEM) improves this machine by adapting its process to the characteristics of the input signal, thereby reducing the sampling rate. Compressed IF-TEM (CIF-TEM) reduces bit usage by performing analog compression before quantization. In this paper, we introduce a combined Adaptive Compressed IF-TEM (ACIF-TEM) -- a new sampler that leverages the two machines, AIF-TEM and CIF-TEM, where each reinforces the effectiveness of the other. We propose an efficient adaptive clockless time-to-digital converter (TDC) architecture for the novel sampler that integrates the compression stage within the TDC, facilitating the realization of the intended integrated system. Via an evaluation study, we demonstrate that the proposed ACIF-TEM sampler achieves lower Mean Square Error (MSE) with fewer bits, offering compression gains of at least 3-bit out of 9-bits over AIF-TEM and 60\% compression over IF-TEM, for fixed recovery MSE with real audio signals.
\end{abstract}

\section{Introduction}\label{sec:intro}
Analog-to-digital converters (ADCs) are fundamental hardware components that enable processing of real-world analog signals by converting them into digital representations. These representations can be stored, compressed, or analyzed using a digital processor \cite{unser2000sampling}. ADCs are used across a wide range of applications, including communication systems \cite{khalili2021mimo}, biomedical data acquisition \cite{ellaithy2023voltage}, and more \cite{mulleti2023power}.

Traditional ADCs perform periodic sampling of analog band-limited (BL) signals at fixed intervals, with a sampling rate that must be at least twice the signal’s maximum frequency \cite{nyquist1928certain}. However, high-rate periodic sampling requires high-frequency clocks \cite{naaman2024hardware}, which suffer from limitations such as increased noise \cite{kinniment1999low} and high power consumption \cite{ingino20014}.  Asynchronous analog-to-digital converters (AADCs) offer an alternative that addresses these limitations \cite{zanjani2021low}.

The Integrate-and-Fire Time Encoding Machine (IF-TEM) is a promising event-based AADC \cite{lazar2004perfect}. Specifically, as illustrated in Fig.~\ref{Fig: ACIF-TEM} in the area outside the dashed boundary, the IF-TEM sampler operates by adding a fixed bias to the input signal and then integrating the result until it reaches a predefined threshold. When the threshold is crossed, the integrator resets and the corresponding time is recorded. Subsequently, these time events are quantized into bits using a time-to-digital converter (TDC) \cite{kalisz2003review}, \cite{li2009delay}, which converts analog time intervals into a bit representation \cite{szyduczynski2023time}.

Unlike classical IF-TEM, a recently proposed Adaptive Integrate-and-Fire Time Encoding Machine (AIF-TEM) adjusts its bias dynamically based on variations in the signal’s energy and frequency \cite{aseel2024AIF}.\off{This enables adaption to changes in local the Nyquist rate of the signal and controls the oversampling rate accordingly.} The adaptive bias is determined from the maximum amplitude observed over a recent time window, allowing the sampler to respond effectively to local signal characteristics, as demonstrated  in the dashed and one dotted block in Fig.~\ref{Fig: ACIF-TEM}. As a result, AIF-TEM improves sampling and quantization efficiency by either reducing the recovery error for a fixed number of samples or lowering the total number of samples required to achieve a given recovery error \cite{aseel2024AIF}.\off{Although IF-TEM sampler offers numerous advantages, it also has certain limitations, such as its unchanging sensitivity to signal variations. To overcome these limitations, the adaptive integrate and fire time encoding machine (AIF-TEM) was proposed in \cite{aseel2024AIF}. Unlike the classic IF-TEM, AIF-TEM sampler dynamically adjusts the bias added to the signal before integration.\off{AIF-TEM sampler differs from classic IF-TEM sampler by changing the bias that is added to the signal prior to the integrator.} This adaptive bias is determined from the maximal amplitude observed over a recent time window. By doing so, AIF-TEM sampler yields a lower oversampling compared to the classical IF-TEM, while minimizing the sampling and quantization distortion.}
To improve the quantization stage of the IF-TEM sampler and enhance storage efficiency, an analog compression-based design, known as Compressed IF-TEM (CIF-TEM), was proposed in \cite{tarnopolsky2022compressed}. 
\begin{figure}
\centering
\includegraphics[width=0.48\textwidth]{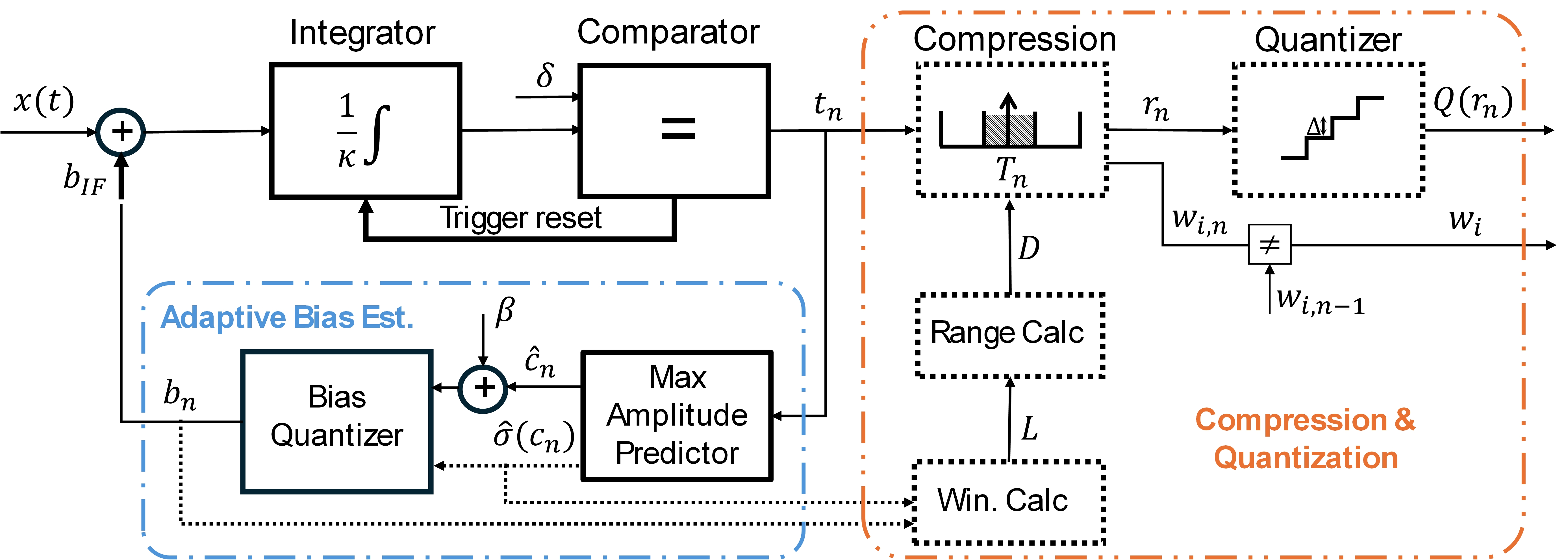}
\caption{\small Adaptive Compressed Integrate-and-Fire Time Encoding (ACIF-TEM). The area outside the dashed boundary represent IF-TEM, while dashed and one dotted block indicate AIF-TEM, and the dashed two dotted block indicate CIF-TEM (See Sec.~\ref{subsec:AIF-TEM}). The dotted lines indicate the new adaptive compression combined approach, detailed in Sec.~\ref{sec:ACIF-TEM}.}
\label{Fig: ACIF-TEM}
\vspace{-0.6cm}
\end{figure}
In this approach, compression, represented by the dashed and two dotted block in Fig.~\ref{Fig: ACIF-TEM}, is performed before quantizing the time intervals, based on their stationarity. This is achieved by dividing a predefined dynamic range into segments, with each sample is quantized and stored based on the segment to which it belongs. Compared to classical IF-TEM, CIF-TEM improves quantization efficiency by significantly reducing the number of bits required to record for a fixed distortion level. 

While AIF-TEM and CIF-TEM enhance distinct modules within IF-TEM, it remains an open question whether their integration yields synergistic benefits or results in performance degradation. Moreover, the proposed CIF-TEM implementation in \cite{tarnopolsky2022compressed} relies on a clock to record time intervals within the fixed dynamic range and perform segment division, which increases power consumption and undermines IF-TEM's advantages over periodic samplers.

In this paper, we propose a novel machine that combines adaptive and compressed designs, as shown in Fig.~\ref{Fig: ACIF-TEM}. In particular, we demonstrate that the adaptive design enhances the compression efficiency of IF-TEM by adaptively reducing the dynamic range of the sampler output. The compression design, in turn, improves the adaptive model by enabling more efficient bias encoding and reducing the overall bit rate in the proposed machine. Moreover, we present an efficient adaptive time-to-digital converter (TDC) architecture that operates without relying on a clock. This TDC architecture is tailored to the combined design, where the compression and quantization are executed within the same system, and utilizing the adaptive dynamic range obtained from AIF-TEM to improve quantization accuracy. Our main contributions are as follows: 
\ifconf 
First, we propose a novel Adaptive Compressed Integrate-and-Fire Time Encoding Machine (ACIF-TEM), which integrates adaptive sampling with an analog compression block, with compression segments dynamically adjusted based on the signal's amplitude. Second, we propose an efficient clockless time-to-digital converter (TDC) architecture which facilitates the reduction of dynamic range tailored to the ACIF-TEM framework.
Finally, we validate ACIF-TEM through numerical evaluations using random band-limited signals and real audio signals. The results show that ACIF-TEM achieves lower Mean Squared Error (MSE) while using fewer bits compared to IF-TEM, AIF-TEM, and CIF-TEM—achieving a reduction of at least 3 bits over AIF-TEM and 60\% compression for IF-TEM for a fixed recovery MSE.
\fi

The paper is organized as follows. Sec.~\ref{sec:Problem} presents the problem formulation, and Sec.~\ref{sec:PRELIMINARIES} covers the essential background. In Sec.~\ref{sec:ACIF-TEM}, we detail the proposed ACIF-TEM. Finally, Sec.~\ref{Evaluation Results} evaluates the performance of ACIF-TEM.

\section{Problem Formulation}\label{sec:Problem}
We address the problem of sampling and reconstructing an analog $c_{max}$-bounded and $2\Omega \text{-BL}$ signal $x(t)$. That is, its amplitude is confined within $|x(t)| \leq c_{max}$ and its Fourier transform is zero for frequencies outside the closed interval $[- \Omega, \Omega]$. We assume the energy $E$ of the signal is finite $E \in \mathbb{R}$ 
\[
\textstyle E = \int_{-\infty}^{\infty} \abs{x(t)}^2 \ dt < \infty.
\]
The relation between the maximum amplitude $c_{max}$ and the bandwidth $\Omega$ is $c_{max} = \sqrt{\frac{E\Omega}{\pi}}$ as considered in \cite{omar2024adaptive,aseel2024AIF}.

After the sampling process, the discrete measurements are quantized to generate bit representations.
The quantized measurements are then decoded to reconstruct the signal $\hat{x}(t)$.
However, sampling in finite regimes and quantization introduces distortion, resulting in a discrepancy between the original input $x(t)$ and the reconstructed signal $\hat{x}(t)$. This distortion is measured using the Mean Squared Error (MSE), expressed in decibels (dB) as 
\begin{equation}\label{MSE}
\textstyle \text{MSE} \triangleq 20 \log_{10} \left( (1/\sqrt{T})\|x(t) - \hat{x}(t)\|_{L_2[0,T]} \right)\quad\text{[dB]}.
\end{equation}
Our goal is to refine the sampling and quantization process to achieve better efficiency in terms of minimizing reconstruction distortion and reducing bit usage.

\section{Background and Related Methods}\label{sec:PRELIMINARIES}
In this section, we provide background on TEM methods and TDC designs that are relevant for the proposed machine.

\subsection{Time Encoding Machines}\label{subsec:AIF-TEM}
Classical IF-TEM \cite{lazar2004time} consists of two core components (See region not enclosed by the dashed line in Fig.~\ref{Fig: ACIF-TEM}): An integrator and a comparator, characterized by a scaling factor $\kappa$ and a threshold $\delta$, respectively. It samples the signal $x(t)$ non-uniformly by capturing specific time instances. First, a fixed bias $b_{\text{IF}}$ is added to the input signal (as shown by the bold arrow in Fig. \ref{Fig: ACIF-TEM}) followed by integration of the product. This bias must be larger than the maximum amplitude of the entire signal, $c_{\text{max}}$, i.e., $b_{\text{IF}} > c_{\text{max}}$, to ensure that the integrator output continuously increases. When the output of the integrator exceeds the threshold $\delta$, the sampling time $t_n$ is recorded, and the integrator is reset. The output of the IF-TEM sampler is a strictly increasing sequence of sampling times ${t_n \mid n \in \mathbb{Z}}$ such that the time differences $T_n$, where $T_n = t_n-t_{n-1}$ correspond to the signal's amplitude.

Adaptive IF-TEM \cite{aseel2024AIF} introduces a dynamic bias $b_n$ that is adjusted at each iteration according to the local amplitude of the signal (See dashed and one dotted block in Fig.~\ref{Fig: ACIF-TEM}). This adaptive approach allows the bias to track changes in signal dynamics, improving the performance in both sampling density and quantization efficiency. Specifically, at each trigger time $t_n$, let $c_n$ denote the maximum absolute amplitude of the signal over a time window spanning the previous $w$ trigger times. The bias adaptor block determines the bias $b_n$ using the estimated maximum amplitude $\hat{c}_n$, provided by the Max Amplitude Predictor (MAP) block, which utilizes the time differences $T_n$. The bias is given by $b_n = \hat{c}_n+\beta$, for $\beta > 0$. AIF-TEM operates correctly as long as the MAP ensures that $b_n \geq c_n$ for all $n$. In practice, the bias values are bounded between $b_{\min}$ and $b_{\max}$, with a fixed number of bits allocated to represent each bias level. As shown in \cite{aseel2024AIF}, the intervals of the firing times differences $T_n$, are bounded for each sample $n$ as follows
\begin{equation}\label{eq:Tn bound aif}
\textstyle \Delta t_{n_{\min}} \triangleq \frac{\kappa \delta}{b_n + c_n} \leq T_n \leq \frac{\kappa \delta}{b_n - c_n}  \triangleq \Delta t_{n_{\max}}.
\end{equation}
We note that the bound in \eqref{eq:Tn bound aif} also applies to IF-TEM by substituting $b_n$ and $c_n$ with their fixed counterparts $b_{\text{IF}}$ and $c_{\max}$. We consider the recovery algorithms for IF-TEM and AIF-TEM as presented in \cite{lazar2004perfect, aseel2024AIF}, which show that a $2\Omega$-bandlimited signal can be perfectly reconstructed from its time samples when $\Delta t_{n_{\max}} < \pi / \Omega$ for all $n$. 

Compressed IF-TEM (CIF-TEM) \cite{tarnopolsky2022compressed} suggested compressing the time outputs, $T_n$, by leveraging signal stationarity. In particular, CIF-TEM adds a window calculation and a compression block to the pipeline (See dashed and two dotted block in Fig.~\ref{Fig: ACIF-TEM}). The Window Calc block, based on the estimated variance, $\hat{\sigma}^2(T_n)$, of the recent $M > 1$ time intervals, computes the number of compression segments $L =  \left \lceil D/{2\hat{\sigma}(T_n)} \right \rceil$, where $D = \Delta t_{n_{\max}} - \Delta t_{n_{\min}}$ is the fixed maximal dynamic range of the time intervals in IF-TEM using $b_{IF}$ and $c_{\max}$ in \eqref{eq:Tn bound aif}. 
The compression block, which is implemented by a counter with a clock in \cite{tarnopolsky2022compressed}, partitions the fixed dynamic range of time intervals, represented by the counter, into these $L$ compression segments. To quantize $T_n$ within its segment, $K$ quantization levels are used. In particular, given a time interval, $T_n$, the compression block determines its corresponding segment index, $w_i$. The higher $\log_2 L$ bits in the counter represent the segment index, while the remaining bits in the counter encode per sample the residual, $r_n$, within that segment. Notably, the compression method leverages the stationarity of the signal by encoding the segment value, $w_i$, only when it changes from the previous sample.
\vspace{-0.1cm}
\subsection{Time-to-digital Converter designs}\label{subsec:TDC}
Unlike clock-based methods \cite{naaman2024hardware,tarnopolsky2022compressed}, TDCs approaches \cite{koscielnik2007designing} can be used in practice to record the firing times differences, $T_n$, efficiently. We note three designs of TDCs that can be identified as relevant for IF-TEM: 1) Analog-based \cite{kalisz2003review}, 2) Delay line \cite{li2009delay}, and 3) Pulse shrinking \cite{szyduczynski2023time}. The focus in this work will be on the third method, as it demonstrates reduced sensitivity to process, voltage, and temperature (PVT) variations compared to Analog-based and Delay line methods \cite{szyduczynski2023time}. An enhancement to this method is the two-step pulse shrinking (2PS) \cite{park2017two}, which enables increased resolution without significantly increasing power or area by separating the measurement into two pulse shrinking measurement modules with fewer components and, as we show in Sec.~\ref{sec:ACIF-TEM}, can be efficiently revised for the adaptive compressed approach proposed herein.

As illustrated in 2PS TDC-Based Compression and Quantization block (the dotted block) in Fig.~\ref{fig:compress_two-step}, 2PS method is based on dividing the measurement into two steps: coarse measurement and fine measurement. The coarse measurement consists of $F$ stages, where in each stage the pulse propagates through a shrinking block that increases the fall time of the time interval $T_{n}$ to a value $\Delta T_{1}$ and cuts the falling edge. This gives a shorter pulse in length $T_{n} - \Delta T_1$, while after $f \leq F$ stages the input pulse $T_n$ disappears. At the last stage before the pulse disappears, the duration of the pulse is $T_{res}$, which satisfies $0 < T_{res} < \Delta T_{1}$. $T_{res}$ is fed into the fine measurement step which includes $G$ stages, each stage similarly shrinks the pulse to the coarse stages in a value $\Delta T_{2}$, where $G\Delta T_{2} = \Delta T_{1}$, provided the pulse fully decays after $g$ fine stages.  
\off{This is the residual of the coarse measurement. 2) The residual is fed into the fine measurement step, which includes $N$ stages. Each stage shrinks the pulse in $\Delta T_{2}$, where $ 0 < N\Delta T_{2} < \Delta T_{1}$. In the same way as the classic pulse shrinking TDC,} 
The number of $f$ coarse stages and $g$ fine stages until the pulse disappears determines the length of the pulse $T_{n}$, where $T_{n} = f\Delta T_{1} + g\Delta T_{2} + Q_{e}$, and $Q_e$ is the quantization error. The main advantage of this method is the wide dynamic range 
$t_{min} = \Delta T_{2} < T_n < F \Delta T_{1} = t_{max}$ while minimizing the sacrifice of area or power consumption in comparison to classic pulse shrinking TDC. We will use this method for the suggested implementation as it performs both compression and quantization in the same module, and enables adaptive dynamic range, as elaborated in Sec.~\ref{sec:ACIF-TEM}.
\vspace{-0.1cm}
\section{Proposed Adaptive Compressed Machine}\label{sec:ACIF-TEM}
In this section, we introduce the Adaptive Compressed Integrate-and-Fire Time Encoding Machine (ACIF-TEM) illustrated in Fig.~\ref{Fig: ACIF-TEM}. The proposed machine integrates the adaptive sampling method of AIF-TEM and the compression strategy CIF-TEM in a novel scheme, which allows each module to enhance the effectiveness of the other, using an effective adaptive TDC implementation tailored for the suggested design. The first stage of the proposed machine comprises applying the adaptive sampling method of AIF-TEM, while the second stage consists of applying a new efficient clockless adaptive compression design. In particular, we introduce here three improvements for the combined approach that significantly enhance performance, compared to AIF-TEM and CIF-TEM presented in Sec.~\ref{sec:PRELIMINARIES}: 1) 
Adjusted segment partitioning during compression based on the amplitude variance, in Sec.~\ref{subsec:BiasPart}. This enables effective bit allocation in the bias vector without the need to store the modified bit lengths in an auxiliary vector.
2) An adaptive dynamic range of the firing time differences using the adaptive bias, in Sec.~\ref{subsec:ATDC}. This results in more accurate quantization that minimizes the MSE of the recovered signal. 3) A new TDC implementation illustrated in Fig.~\ref{fig:compress_two-step} that enables adaptive dynamic range adaptation and performs the compression and quantization in one effective module, based on the 2PS method in Sec.~\ref{subsec:TDC}. In the following, we provide a detailed description of the proposed method.

\begin{figure}
    \centering
    \includegraphics[width=1\linewidth]{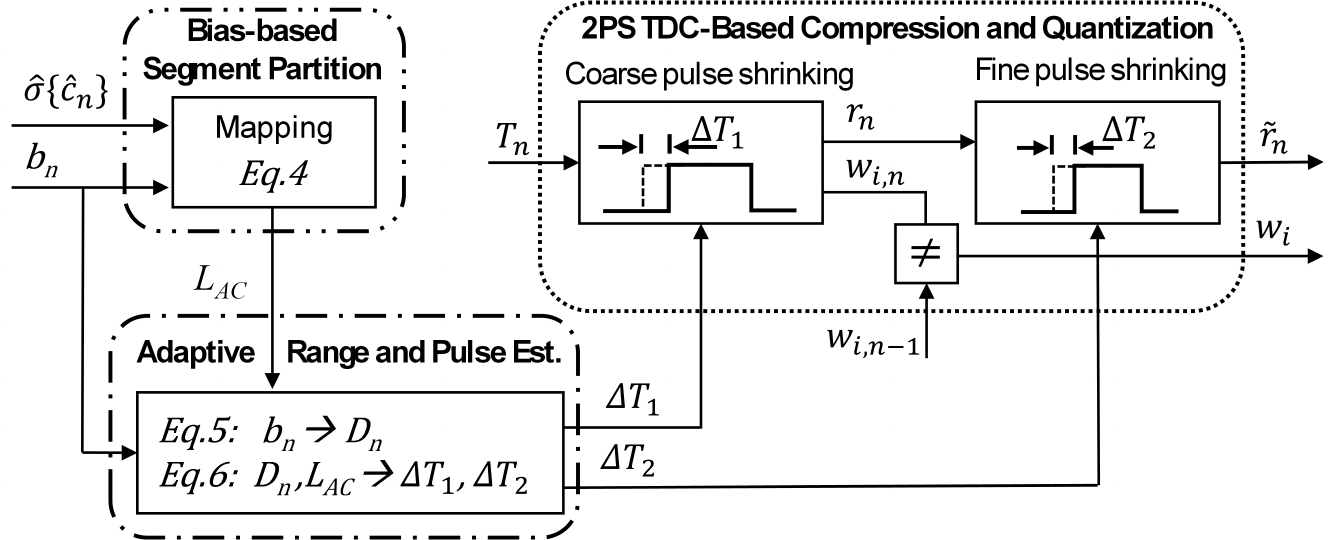}
    \caption{\small ACIF-TEM compression design using clock-less 2PS TDC. The coarse and fine pulse shrinking blocks in 2PS TDC block replace the compression block and the quantizer block in Fig.~\ref{Fig: ACIF-TEM}, respectively.}
    \label{fig:compress_two-step}
    \vspace{-0.6cm}
\end{figure}

\vspace{-0.2cm}
\subsection{Bias-based Segment Partition}\label{subsec:BiasPart}
\off{
\begin{itemize}
    \item the first part is adjusting the number of segments in a way that enables enhancement in the bias vector and save estimation
    \item the enhancement in the bias vector is adjusting the number of bits etc.
    \item to avoid saving the changing number of bits we define the number of segments according to the bias vector
    \item instead of using the variance of Tn we use the variance of bn, due to the following relation
    \item in this manner, win calc block includes only estimator (rather than mapping), cause we use bn estimator.
\end{itemize}
}
In this section, we present the first block of the suggested machine (see Bias-based Segment Partition block in Fig. \ref{fig:compress_two-step}), which adjusts the number of compression segments $L_{AC}$ in a way that facilitates bias vector enhancement without increasing the total number of bits, and eliminates the need for an estimator in the compression block. This enhancement in the bias vector is achieved by adaptively varying the number of bits, so that the number of bits assigned to encode the bias corresponds to the variability of the signal, in contrast to the classic AIF-TEM implementation that uses fixed number of bits $R_b$. Specifically, we dynamically adapt the number of quantization levels $L_b$ according to the estimated amplitude variance $\hat{\sigma}^2(c_n)$ from the MAP block as shown in the dotted arrow in Fig. \ref{Fig: ACIF-TEM}, where the number of bias quantization levels are $L_b = \left\lceil \frac{c_{max}}{\hat{\sigma}(c_n)}\right\rceil$. In this manner, when signal variations are small, finer segmentation allows for more precise bias representation. In contrast, high variability leads to coarser bias encoding with fewer bits, which leads to less frequent updates, making this method more suitable for real-time sampling.

To avoid recoding the changes in the number of bits allocated for $R_b$, we define the number of compression segments $L_{AC}$ (equivalent to $L$ in CIF-TEM) as a function the bias $b_n$ and his estimated variance $\hat{\sigma}(b_n)$, that equals to $\hat{\sigma}(c_n)$ (see dotted-line arrow in Fig.~\ref{Fig: ACIF-TEM} and Bias-based Segment Partition block in Fig.~\ref{fig:compress_two-step}). In particular, based on the first-order Taylor series approximation, the relationship between $Var(X)$ and $Var(f(X))$ can be expressed as follows $Var[f(x)]\approx (f'(\mathbb E[X]))^2 Var[X]$, applying this to \eqref{eq:Tn bound aif} yields the following relationship between $\sigma^2(T_n)$ and $\sigma^2(c_n)$
\begin{equation}
\label{eq:var_Tn_cn}
    \textstyle  \sigma^2(T_n) \geq  \frac{(2\kappa \delta)^2}{(2 \mathbb E[c_n]+\beta)^4} \cdot\sigma^2(c_n). 
\end{equation}

In practice, an additional constraint on the number of compression segments arises due to the finite number of quantization levels, $K$. The total number of bits allocated for both the compression segments $L_{AC}$ and the residual $r_n$ is denoted by $N_b = \log(K)$. Since the residual requires at least one bit for representation, the maximum number of bits available for the compression segments is $N_b-1$. This leads to a bound on the number of compression segments given by  $L_{AC}\leq \frac{K}{2}$.  \off{In practice, the maximum number of compression segments is constrained by the number of quantization levels in the $K$-level uniform quantizer, such that the total bit budget satisfies  $N_b = \log(K)$. Considering that at least one bit must be reserved for the residual $r_n$ \textcolor{red}{to ensure that the sampler produces at least one bit of output, }the maximum number of bits allocated to the compression segment is $\eta = N_b - 1$, \textcolor{red}{confirming that sampling occurred at that specific time instant,} .} The number of quantization levels for the bias is $L_b = \left\lceil {c_{max}}/{\hat{\sigma}(c_n)}\right\rceil$ and the number of compression segments  in CIF-TEM is $L = \left \lceil D/{2\hat{\sigma}(T_n)} \right \rceil = \frac{2\kappa \delta}{\beta(2c_{max}+\beta)}\cdot \frac{c_{max}}{2\sigma(T_n)}$. Applying in $L$ the lower bound of $\sigma(T_n)$ in \eqref{eq:var_Tn_cn} and $L_b$, we obtain an upper bound on the number of compression segments
\vspace{-0.3cm}
\begin{multline}
\label{number_of_seg}
    \textstyle  L_{AC} \triangleq \min\left[\frac{\kappa \delta}{\beta(2c_{max}+\beta)}\cdot \frac{c_{max}}{ \sigma(c_n)}\frac{(2\mathbb{E}[c_n]+\beta)^2}{2\kappa \delta},\frac{K}{2}\right]\\
    = \min\left[\phi_m(c_n) \cdot L_{b},\frac{K}{2}\right],
    \vspace{-0.3cm}
\end{multline}
where $\phi_m(c_n) = \frac{(2 \mathbb E[c_n]+\beta)^2}{2 \beta(2 c_{max}+\beta)}$, and $m=\left \lfloor\frac{n}{M}\right\rfloor$ represents the current count of updates to the compression segments. 
In this manner, the number of bias bits satisfies $R_b = \log_2(L_{AC})-\log_2(\phi_m(c_n))$, where $L_{AC}$ is provided by the compression block.
It is important to note that, this method utilizes the estimator from the MAP block, eliminating the need for an estimator in compression and quantization blocks. Thus, the Window calc block performs only mapping rather than additional estimation (see Win calc block in Fig.~\ref{Fig: ACIF-TEM} and Bias-based Segment Partition block in Fig.~\ref{fig:compress_two-step}). 

\vspace{-0.2cm}
\subsection{Adaptive Dynamic Range using Time-to-Digital Converter}\label{subsec:ATDC}
The second stage of the proposed machine consists of two main building blocks (See Fig.~\ref{fig:compress_two-step}): 1) The first one, adaptively estimates the actual dynamic range of the time interval $T_n$, and adjusts $\Delta T_{1}$ (the coarse pulse shrinking) and $\Delta T_{2}$ (the fine pulse shrinking) to enhances the MSE performance of the TDC-based analog compression and quantization suggested next (See Adaptive range and pulse Est. block in Fig. \ref{fig:compress_two-step}). 2) The second one performs adaptive analog compression and quantization efficiently at a single module using an adjusted compact TDC design that only recodes  $r_n$, the dynamic part of $T_n$, per sample with a smaller number of bits (See 2PS TDC-Based Compression and Quantization block in Fig. \ref{fig:compress_two-step}).  

In particular, the estimation of the adaptive dynamic range at the first proposed main block is defined as 
\begin{equation}
   \textstyle D_n = \Delta t_{n_{max}}-\Delta t_{n_{min}},
\end{equation}
where $\Delta t_{n_{min}} = \frac{\kappa\delta}{b_n+\hat{c}_n} = \frac{\kappa\delta}{2b_n-\beta}$ and $\Delta t_{n_{max}} = \frac{\kappa\delta}{b_n-\hat{c}_n} = \frac{\kappa\delta}{\beta}$ for $b_n=\hat{c}_n+\beta$. We note: 1) that $D_n\leq D$ as $b_n\leq b_{max} = b_{IF}$ and $\hat{c}_n\leq c_{max}$, and 2) that in practice only the recoded $b_n$ is needed to compute the actual $D_n$.

For the compression and quantization with adaptive dynamic range, a clock-less two-step pulse-shrinking TDC is utilized. As described in Sec.~\ref{subsec:TDC}, 2PS TDC divides the measured time intervals into segments using the coarse shrinking steps, and quantize the residual using the fine shrinking steps. In the suggested machine, we utilized the $F$ coarse shrinking steps, denoted by coarse shrinking block in Fig.~\ref{fig:compress_two-step}, for the segment separation in the compression method (compression block in Fig.~\ref{Fig: ACIF-TEM}), that is $F = L_{AC}$ (See \eqref{number_of_seg}). In addition, we use the $G$ fine shrinking steps, denoted by the fine shrinking block in Fig.~\ref{fig:compress_two-step}, as a quantizer for the residual \textcolor{blue}{$r_n$} (quantizer in Fig.~\ref{Fig: ACIF-TEM}), that is $G = \log \frac{K}{L_{AC}}$. As a result, we obtain,
\begin{equation}
\label{eq:two_step}
\begin{gathered}
    \textstyle \Delta T_{1} = {D_n}/{L_{AC}} \quad \text{and} \quad
    \Delta T_{2} = {\Delta T_{1}}/{\log_2(\frac{K}{L_{AC}})},
    \end{gathered}
\end{equation}
where the low and finite resolutions vary with $D_n$ and $b_n$.

Now, given a time interval, $T_n$, the coarse shrinking block determines its corresponding segment index, $w_{i,n}$, while the fine shrinking block determines the quantized residual, $\tilde{r}_n$, encoded per sample. To reduce bit usage by the adaptive compressor, the segment value $w_i$ is recorded only when it changes from the previous sample, i.e., when $w_{i,n}\neq w_{i,n-1}$.

\section{Evaluation Results}\label{Evaluation Results}
\off{In finite regime, for fixed MSE we receive significantly lower number of bits.
}
We validated ACIF-TEM performance using both synthetic MATLAB signals and real audio signals. The synthetic input signals are $2 \Omega$-BL $c$-bounded signals, where $\Omega =2\pi \cdot10Hz$. The simulated signals are given by
\begin{equation*}
    \label{eq: sim_signal}
   \textstyle x(t) = \sum_{n =-M}^{M} \frac{sin(\Omega (t-n \frac{\pi}{\Omega}))}{\Omega(t-n \frac{\pi}{\Omega})},
\end{equation*}
where $M = 2$, $t \in \mathbb{R}$ \off{and $a[n]$ is randomly selected 200 times from [-1,1]}. IF-TEM, CIF-TEM, AIF-TEM, and ACIF-TEM samplers are configured with the parameters $\kappa = 0.24$, $\delta = 0.0156$. The fixed bias in IF-TEM and CIF-TEM is configured as $b = 3.4166$. In AIF-TEM and ACIF-TEM the amplitude $c_{n}$ is estimated in the MAP block as $\hat{c}_{n} = \alpha_{1}z_{n}+(1-\alpha_{1})\hat{c}_{n-1}$, where $z_{n}$ is given by $z_{n} = -b_{n} + \kappa \delta / T_{n}$, while $\alpha_{1} = 0.98$. The predicted $c_{n+1}$ is $\hat {c}_{n+1} = \hat{c}_{n} + \alpha_{2} s_{n}$ where $\alpha_{2} = 0.6$, and $s_{n}$ is the standard deviation of $\hat{c}_{n}$, and computed via Welford's method \cite{welford1962note}. 

Fig.~\ref{fig:MSE_bits}(a) presents the sampling distortion in terms of MSE as defined in \eqref{MSE}, for a fixed number of bits (see upper horizontal axis). The performance of classic IF-TEM is indicated by the blue line, the red line marks the AIF-TEM sampler. CIF-TEM is denoted by the purple line, and the orange line marks the suggested ACIF-TEM sampler. IF-TEM and CIF-TEM exhibit a lower slope compared to AIF-TEM and ACIF-TEM, owing to the smaller number of samples in the latter.  The MSE obtained using ACIF-TEM sampler is significantly lower than the MSE for AIF-TEM and CIF-TEM. ACIF-TEM also yields better performance than a simple combination of AIF-TEM and CIF-TEM in -10dB. 

The total number of bits in Fig.~\ref{fig:MSE_bits}(b), showing the reduction in the number of bits when using ACIF-TEM sampler for fixed MSE (as indicated on the upper horizontal axis).
\begin{figure}
    \centering
    \includegraphics[width=1\linewidth]{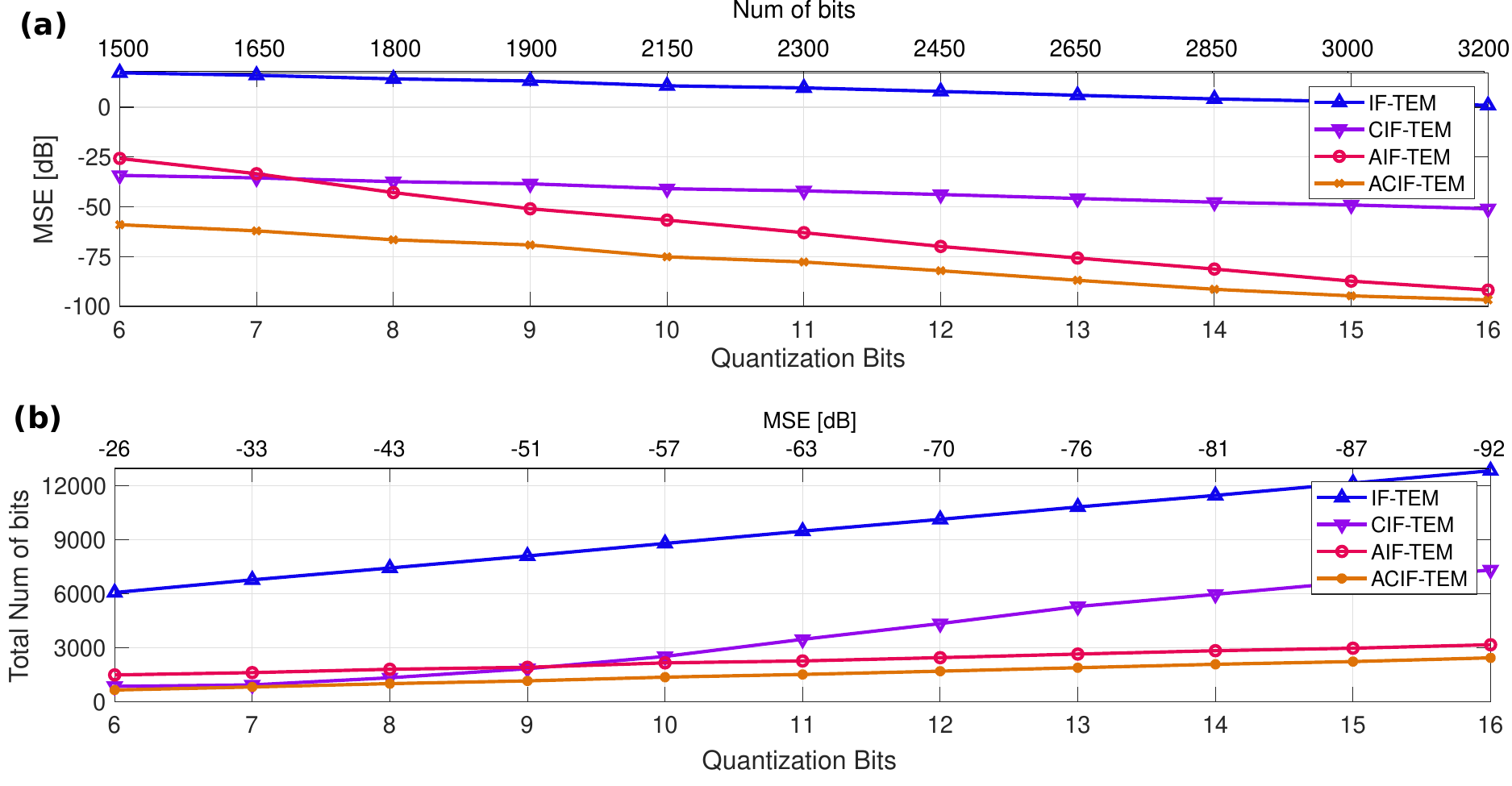}
    \caption{\small Performance comparison between IF-TEM, AIF-TEM, CIF-TEM, and ACIF-TEM in terms of (a) MSE in dB for fixed bits and (b) total number of bits for fixed MSE for all the samplers.}
    \vspace{-0.6cm}
    \label{fig:MSE_bits}
\end{figure}

For the same Mean Squared Error (MSE), the compression percentage achieved by ACIF-TEM sampler for a stationary synthetic sinc signals is approximately  40\% and 80\% relative to AIF-TEM and the classic IF-TEM, respectively, while the  compression percentage for a quasi-stationary real audio signals is approximately 30\% and 60\% with respect to AIF-TEM and IF-TEM. 
\off{

\begin{equation}
    A = \frac{1}{2}
\end{equation}
 \begin{figure}
    \centering
    \includegraphics[width=1\linewidth]{total_num_of_bits_31_7_fixed.eps}
    \end{figure} 
\begin{figure}
    \centering
    \includegraphics[width=1\linewidth]{total_num_of_bits_31_7.eps}
    \end{figure}
 \begin{figure}
    \centering
    \includegraphics[width=1\linewidth]{MSE_31_7.eps}
    \end{figure}   
 }
 \vspace{-0.2cm}
\section*{Acknowledgments}
The authors thank Alejandro Greiver and Shaked Peleg, for their contributions to the simulations in Sec.~\ref{Evaluation Results}.

\bibliographystyle{IEEEtran}
\bibliography{ref}

\end{document}